\documentclass[10pt,a4paper]{article}

\usepackage{url}
\usepackage[english]{babel}
\usepackage[utf8]{inputenc}
\usepackage[T1]{fontenc}
\usepackage{graphicx}
\usepackage{geometry}
\usepackage{tabularx}
\usepackage{psfrag}
\usepackage{fancyhdr}
\usepackage{xcolor}
\usepackage{sistyle}
\usepackage{amsfonts}
\usepackage{amsmath}
\usepackage{amssymb}
\usepackage{dsfont}
\usepackage{booktabs}
\usepackage{multirow}
\usepackage{amsthm}
\usepackage{caption}
\usepackage{subcaption}
\usepackage{pifont}
\usepackage[ruled]{algorithm2e}
\usepackage{enumerate}
\usepackage{enumitem}
\usepackage[numbers, comma, sort&compress]{natbib}
\usepackage{bbold}
\usepackage{bbm}
\usepackage[bottom]{footmisc}
\usepackage{tikz}
	\usetikzlibrary{backgrounds,automata,shapes}
	\usetikzlibrary{decorations.pathreplacing, arrows}
\usepackage{hyperref}
\usepackage{mathrsfs}
\usepackage{bigints}
\usepackage{xr}

\theoremstyle{theorem}

\newtheorem{lemma}{Lemma}
\newtheorem{theorem}{Theorem}
\theoremstyle{definition}
\newtheorem{definition}{Definition}

\theoremstyle{remark}
\newtheorem{remark}{Remark}

\newcommand{\e}{\mathrm{e}}

\newcommand{\GAF}{\mathsf{GAF}}
\newcommand{\wn}{\zeta}
\newcommand{\dd}{\mathrm{d}}
\newcommand{\Ell}{L^2}
\newcommand{\ii}{\mathrm{i}}

\graphicspath{{figures/}}

\begin{document}

\title{The Analytic Stockwell Transform and its Zeros}

\author{Ali Moukadem\thanks{A. Moukadem, J.-B. Courbot and N. Juillet are with IRIMAS UR 7499, Université de Haute-Alsace, Mulhouse, France (e-mail: firstname.lastname@uha.fr);},
Barbara Pascal\thanks{B. Pascal is with Nantes Université, École Centrale Nantes, CNRS, LS2N, UMR 6004, F-44000 Nantes, France, (e-mail: barbara.pascal@cnrs.fr).}, Jean-Baptiste Courbot\footnotemark[1], Nicolas Juillet\footnotemark[1]
}%

\maketitle

\abstract{The Stockwell Transform is a time--frequency representation resulting from an hybridization between the Short-Time Fourier Transform and the Continuous Wavelet Transform. 
Instead of focusing on energy maxima, an unorthodox line of research has recently shed the light on the zeros of time--frequency transforms, leading to fruitful theoretical developments combining probability theory, complex analysis and signal processing. 
While the distributions of zeros of the Short-Time Fourier Transform and of the Continuous Wavelet Transform of white noise have been precisely characterized, that of the Stockwell Transform of white noise zeros remains unexplored.
To fill this gap, the present work proposes a characterization of the distribution of zeros of the Stockwell Transform of white noise taking advantage of a novel generalized Analytic Stockwell Transform.
First of all, an analytic version of the Stockwell Transform is designed.
Then, analyticity is leveraged to establish a connection with the hyperbolic Gaussian analytic function, whose zero set is invariant under the isometries of the Poincaré disk.
Finally, the theoretical spatial statistics of the zeros of the hyperbolic Gaussian analytic function and the empirical statistics of the zeros the Analytic Stockwell Transform of white noise are compared through intensive Monte Carlo simulations, showing the practical relevance of the established connection.
A documented Python toolbox has been made publicly available by the authors.
}

\vspace{3mm}

\vspace{-2mm}
\tableofcontents

\section{Introduction}
\label{sec:intro}

Time--frequency analysis is a cornerstone of nonstationary signal processing~\cite{flandrin1998time,Grochenig2001,flandrin2018explorations}.
Crucially, contrary to the Fourier spectrum describing the overall frequency content but loosing time localization, time--frequency representations aim at accounting for both the temporal dynamics and the time-varying frequency content of a signal.
Among the most widely used time--frequency representations are the \emph{spectrogram}, obtained as the squared modulus of the Short-Time Fourier Transform~\cite{flandrin1998time,flandrin2018explorations}, and the time--scale \emph{scalogram}, which is the squared modulus of the Wavelet Transform~\cite{stephane1999wavelet}.
These joint representations appeared particularly well-suited to analyze physical phenomena as diverse as bat echolocation~\cite{flandrin1988time}, earthquakes~\cite{wu2017s} and gravitational waves~\cite{chassande2006best}.

The search for \emph{analytic} signal transforms is a long-standing problem both in signal processing and harmonic analysis~\cite{Ascensi2009,mallat2015phase,waldspurger2017phase,Holighaus2019,Bardenet2021,Pascal2022} as analyticity is a key property in phase retrieval problems and phaseless reconstruction~\cite{mallat2015phase,waldspurger2017phase,Holighaus2019}.
Recently, this interest has been strongly renewed by an unorthodox research path which emerged at the interface between time--frequency analysis and stochastic geometry.
This novel line of research consists in shifting the historical interest in energy maxima toward \emph{zeros}, which can be seen as ``silent'' points in the time--frequency plane~\cite{gardner2006sparse}.
Notably, the seminal works~\cite{gardner2006sparse,Flandrin2015} have shed light on the remarkably even distribution of the zeros of the spectrogram of  \emph{white noise}.
Not only this observation received a rigorous mathematical explanation~\cite{Bardenet2020} elegantly combining complex analysis and probability theory, but the established link between time--frequency analysis and random complex analytic functions opened up a novel line of research in the theory of signal representation and led to the characterization of the zeros of many other signal representations, among which the Wavelet Transform~\cite{Abreu2018,Koliander2019,Bardenet2020,Bardenet2021}.
Further, this fruitful connection has even motivated the design of new representations, and in particular of discrete representations with desirable algebraic properties~\cite{Bardenet2021,Pascal2022}, as well as original contributions in harmonic analysis~\cite{abreu2022local}.

The present work aims at complementing these studies by characterizing the distribution of the zeros of the time--frequency \emph{Stockwell Transform}~\cite{Stockwell1996}, which genuinely hybridizes the Short-Time Fourier Transform and the Wavelet Transform so as to provide an original multi-resolution analysis with a local phase having an absolute reference~\cite{Stockwell1996,Stockwell2007,Riba2015}.
The motivation for the introduction of the Stockwell transform was to overcome a major practical drawback of the Short-Time Fourier Transform: for low (resp. high) frequency the frequency (resp. time) resolution is coarse, impairing precise localization in the time-frequency plane.
To that aim the historical design relies on a Gaussian analysis window~\cite{Stockwell1996}~\cite[Section 3.3]{shah2022wavelet}, whose width is varied as a function of frequency so as to produce a higher frequency resolution at lower frequencies and a higher time resolution at higher frequencies while preserving the absolute reference of the phase, which is lost when using the Wavelet Transform.
For a comprehensive and pedagogical presentation of the Wavelet and Stockwell Transforms, and links between both, the reader is referred to~\cite{stephane1999wavelet} and to the more recent textbook~\cite{shah2022wavelet}.
Several extensions of the Stockwell Transform have been proposed in the literature~\cite{Pinnegar2003,Battisti2016,liu2019self, Moukadem2015, Zidelmal2017}, to cite but a few. The main objective of these extensions is to enable to adapt closely to the characteristics of the analyzed signal and to the time--frequency patterns of interest.
Thanks to its solid theoretical ground, high versatility and ease of interpretation, the Stockwell Transform has progressively gained popularity not only in signal processing~\cite{Wei2021, Huang2021,Wei2022}, but also in applied mathematics~\cite{Riba2015,Shah2021} and mathematical physics~\cite{Guo2010,Hutnikova2015}.

The main contributions of the present work are, first, to construct an \emph{Analytic} Stockwell Transform; second, thanks to a rigorous construction of an adapted continuous white noise, to characterize the zeros of the Analytic Stockwell Transform of white noise by proving that they coincide with the zeros of the so-called \emph{hyperbolic Gaussian analytic function} hence demonstrating that the zeros are invariant under isometries of the Poincaré disk; third, to provide numerical evidence supporting this connection by comparing through intensive Monte Carlo simulations the theoretical \emph{pair correlation function} of the zeros of the hyperbolic Gaussian analytic function and the empirical pair correlation function of the zeros of the Analytic Stockwell Transform of white noise.
In Section~\ref{sec:GST}, after reminding the definition of the Generalized Stockwell Transform, the existence of an Analytic Stockwell Transform is proven by leveraging the intrinsic link between the Stockwell and Wavelet transforms.
Then, the distribution of the zeros of the Analytic Stockwell Transform of white noise is characterized in Section~\ref{sec:zeros}, taking advantage of a well-suited construction of white noise and of a connection with Gaussian analytic functions.
Finally, numerical experiments are implemented in Section~\ref{sec:numeric}, providing illustrations of the zero pattern of the Analytic Stockwell Transform and quantitative support to the results of Section~\ref{sec:zeros} by systematically comparing the theoretical and empirical spatial statistics of zeros in hyperbolic geometry through intensive Monte Carlo simulations.

\noindent \textbf{Notation.} The set of real (resp. complex) numbers is denoted by $\mathbb{R}$ (resp. $\mathbb{C}$), while the set of positive real numbers (resp. complex number with positive imaginary part) is denoted by $\mathbb{R}_+$ (resp. $\mathbb{C}_+$).
The set of nonnegative (resp. positive) integers is denoted by $\mathbb{N}$ (resp. $\mathbb{N}^*$).
The Hilbert space of functions which are square-integrable with respect to the Lebesgue measure on $\mathbb{R}$ equipped with the scalar product $\langle f, g \rangle = \int_{\mathbb{R}} f(t) \overline{g(t)} \, \dd t$ is denoted by $\Ell(\mathbb{R})$.
The space of integrable functions is denoted by $L^1(\mathbb{R})$.
For $f \in \Ell(\mathbb{R})\cap L^1(\mathbb{R})$, the \emph{Fourier transform} of $f$, denoted by $\widehat{f}$, is defined as: $\forall  \nu \in \mathbb{R}, \, \widehat{f}(\nu) = (2\pi)^{-1/2}\int_{\mathbb{R}} f(t) \e^{-2\pi \ii \nu t} \, \dd t$.
The standard \emph{complex} Gaussian distribution of zero mean and unit variance is denoted by $\mathcal{N}_{\mathbb{C}}(0,1)$.
The closed disk of radius $R > 0$ is denoted by $\overline{\mathsf{B}}(0,R) = \lbrace z \in \mathbb{C} : \lvert z \rvert < R \rbrace$.
For $\Lambda \subseteq \mathbb{C}$, the space of analytic functions on $\Lambda$ is denoted by $\mathcal{A}(\Lambda)$.

\section{The Analytic Stockwell Transform}
\label{sec:GST}

First, the formal definition of the Generalized Stockwell Transform is reminded.
Then, connecting the Stockwell Transform to the Wavelet Transform and leveraging analyticity results obtained for Cauchy wavelet, an Analytic Stockwell Transform is constructed.

\subsection{Preliminary: definition of the Generalized Stockwell Transform}

\begin{definition}
Let $\varphi \in \Ell(\mathbb{R})$. 
The \emph{Generalized Stockwell Transform} with analysis window $\varphi$ of 
$f \in \Ell(\mathbb{R})$ is defined as
\begin{equation} 
\label{eq:def_GST}
    S_\varphi f : \left\lbrace
\begin{array}{rcl}
     \mathbb{R} \times \mathbb{R}_+ & \rightarrow & \mathbb{C}  \\
     (x, \xi) & \mapsto & \xi \int_{-\infty}^\infty f(t)  \overline{\varphi (\xi(t-x))} \e^{-2\pi \ii t\xi} \, \dd t. 
\end{array}
    \right.
\end{equation}
\end{definition}

\begin{remark}
The seminal Stockwell transform introduced in~\cite{Stockwell1996} corresponds to the standard Gaussian analysis window $\varphi(t) = (2\pi)^{-1/2} \e^{-t^2/2}$.
\end{remark}

For any $x \in \mathbb{R}$, $\xi,\gamma \in \mathbb{R}_+$, let the \emph{translation}, \emph{modulation} and \emph{dilatation} automorphisms of $\Ell(\mathbb{R})$ be respectively defined as
\begin{align}
\begin{split}
    \mathbf{T}_xh(t) = h(t-x); \quad
    \mathbf{M}_\xi h(t) = \e^{-2\pi \ii \xi t} h(t);  \quad
    \mathbf{D}_\gamma h(t) = \gamma h(\gamma t),
\end{split}
\end{align}
Then, the Generalized Stockwell Transform of Equation~\eqref{eq:def_GST} can be interpreted as the decomposition of the signal onto family of functions generated by translations, modulations and dilatations of the analysis window, that is:
 \begin{equation}
 \label{eq:stockwell1}
   S_\varphi f(x,\xi)  =\langle  f, \mathbf{M}_\xi \mathbf{T}_x \mathbf{D}_{\xi} \varphi \rangle.
 \end{equation}
 
Note the modulation operator $\mathbf{M}_\xi$, which essentially differentiates the Stockwell transform from the Wavelet transform and provides an absolute reference for the phase information \cite{Stockwell1996}.

\subsection{Existence of an Analytic Stockwell Transform}
\label{ssec:AST}

The design of an Analytic Stockwell Transform is based on the \emph{Cauchy} wavelet, also referred to as the \emph{Daubechies-Paul} wavelet in the literature~\cite{daubechies1988time}.

\begin{definition}
Let $\beta \in \mathbb{R}$ be such that  $\beta > 0$.
The \emph{Cauchy} wavelet of parameters $\beta$ is defined as the unique function $\psi_{\beta} \in L^2(\mathbb{R})$ whose Fourier transform reads
\begin{align}
\label{eq:def-cauchy}
    \widehat{\psi}_{\beta}(\nu) = 
    \left\lbrace 
    \begin{array}{cl}
    \nu^{\beta}\e^{-2\pi \nu} & \text{if } \nu \geq 0 \\
    0 & \text{otherwise.}
    \end{array}
    \right.
\end{align}
\end{definition}

\begin{remark}
The condition $\beta > 0$ ensures the admissibility of the wavelet, i.e.,
\begin{align}
\int_{\mathbb{R}} \left\lvert \widehat{\psi}_\beta(\nu) \right\rvert^2 \, \frac{\mathrm{d}\nu}{\lvert \nu \rvert } < \infty
\end{align}
and that $\widehat{\psi}_{\beta} \in L^2(\mathbb{R})$.
Then, the bijectivity of the Fourier transform from $L^2(\mathbb{R})$ to $L^2(\mathbb{R})$ yields the existence of a unique $\psi_{\beta}$ satisfying~\eqref{eq:def-cauchy}.
\end{remark}

\begin{theorem}
\label{thm:AST}
    Let $\beta > 0$ and $\psi_{\beta} \in L^2(\mathbb{R})$ be the Cauchy wavelet of parameter $\beta$. Define the \emph{modulated} Cauchy wavelet $\varphi_{\beta}(t) = \psi_{\beta}(t) \e^{-2\pi \ii t}$.
     There exists a smooth nonvanishing function $\lambda : \mathbb{R} \times \mathbb{R}_+ \rightarrow \mathbb{C}$ and a conformal mapping $ \vartheta : \mathbb{C}_+ \rightarrow \mathbb{D}$ such that any signal $f \in \Ell(\mathbb{R})$ can be written as
    \begin{align}
    \label{eq:S-lambda-W}
       \forall (x,\xi) \in \mathbb{R} \times \mathbb{R}_+, \quad S_{\varphi_{\beta}} f(x,\xi) = \lambda(x,\xi) \times F\left(\vartheta(x + \ii \xi^{-1})\right)
    \end{align}
    with $F : \mathbb{D} \rightarrow \mathbb{C}$ an \emph{analytic} function on the unit disk.
\end{theorem}

In words, Theorem~\ref{thm:AST} states that, up to multiplication by a smooth nonvanishing factor and composition with a conformal transformation, the Generalized Stockwell Transform of analyzing window $\varphi_{\beta}$  of any finite energy signal is an analytic function on the unit disk $\mathbb{D}$.

\begin{proof}
First, remark that for $f\in\Ell(\mathbb{R})$, $\forall (x,\xi) \in \mathbb{R}\times\mathbb{R}_+$
\begin{align}
\begin{split}
\label{eq:AST-to-CWT}
     S_{\varphi_{\beta}} f(x,\xi) &= \xi \int_{-\infty}^\infty f(t) \overline{\varphi_\beta(\xi(t-x))} \e^{-2\pi \ii \xi t} \, \dd t\\
     &= \xi \e^{- 2\pi \ii \xi x} \int_{-\infty}^\infty f(t) \overline{\varphi_\beta(\xi(t-x)) \e^{2\pi \ii \xi (t - x)}} \, \dd t\\
     &= \xi \e^{- 2\pi \ii \xi x} \int_{-\infty}^\infty f(t) \overline{\psi_\beta(\xi(t-x) )} \, \dd t\\
     &=\sqrt{\xi} \e^{- 2\pi \ii \xi x} \frac{1}{\sqrt{\xi^{-1}}} \int_{-\infty}^\infty f(t) \overline{\psi_\beta\left(\frac{t-x}{\xi^{-1}} \right)} \, \dd t\\
     & = \sqrt{\xi} \e^{- 2\pi \ii \xi x} W_{\psi_\beta} f(x, \xi^{-1})
     \end{split}
\end{align}
where $W_{\psi_{\beta}} f$ is the Cauchy \emph{Wavelet} Transform of $f$.

Then, from~\cite[Section 2.2]{mallat2015phase} and ~\cite[Theorem~2]{Holighaus2019}, the Cauchy Wavelet Transform of $f \in \Ell(\mathbb{R})$ is analytic on the upper-half complex plane, that is $G : z = x+\ii y \mapsto y^{-\beta-1/2}W_{\psi_{\beta}} f(x,y)$    is an analytic function on $\mathbb{C}_+$.
Moreover, let $\vartheta : \mathbb{C}_+ \rightarrow \mathbb{D}$ be the \emph{Cayley} transform defined as 
\begin{align}
\label{eq:Cayley}
\vartheta(z)= \frac{z - \ii}{z + \ii}.
\end{align}
Then, $\vartheta$ is a conformal and invertible mapping from the upper-half plane to the unit disk whose inverse is also conformal and writes 
\begin{align}
\vartheta^{-1}(w) = \frac{\ii + \ii w}{1 - w}.
\end{align}
Define $F : \mathbb{D} \rightarrow \mathbb{C}$ as
$\forall w \in \mathbb{D}, \, F(w) = G(\vartheta^{-1}(w))$.
Since $F$ is the composition of an analytic function $G$, and of a conformal, hence holomorphic, mapping $\vartheta^{-1}$, $F$ is analytic.
Finally, let $\lambda(x,\xi) = \xi^{-\beta} \e^{- 2\pi \ii \xi x} $, then $\lambda : \mathbb{R} \times \mathbb{R}_+ \rightarrow \mathbb{C}$ is a smooth nonvanishing function.

By construction of $\vartheta$, $\lambda$ and $F$ they satisfy all the conditions enumerated in Theorem~\ref{thm:AST} and $\forall (x,\xi) \in \mathbb{R} \times \mathbb{R}_+, \, S_{\varphi_{\beta}} f(x,\xi) = \lambda(x,\xi) \times F\left(\vartheta(x + \ii \xi^{-1})\right)$ which concludes the proof.
\end{proof}

\section{The Zeros of the Analytic Stockwell Transform of White Noise}
\label{sec:zeros}

The purpose of this section is to characterize the probability distribution of the zeros of the newly introduced Analytic Stockwell Transform of white noise.
To this purpose, the link between the Analytic Stockwell Transform and the Cauchy Wavelet Transform, established in the proof of Theorem~\ref{thm:AST}, and the one-to-one correspondence between the zeros of the Cauchy Wavelet Transform of white noise and the zeros of the so-called \emph{hyperbolic Gaussian analytic function}, proven in~\cite{Koliander2019,Bardenet2021}, are combined to yield Theorem~\ref{thm:AST-GAF} below.

For the sake of completeness, before proving the connection with the hyperbolic Gaussian analytic function,
some necessary elements of hyperbolic and stochastic geometry are presented; the definition of the hyperbolic Gaussian analytic function is reminded accompanied with its main interesting properties so that the present section is self-contained.
The interested reader can refer to~\cite{Hough2009} for a thorough presentation of the theory of Gaussian analytic functions and a detailed study of their zeros.

\subsection{Gaussian analytic functions and point processes}
\label{ssec:GAF-PP}

The \emph{hyperbolic plane} is the Riemaniann surface, i.e, two-dimensional manifold, that is homeomorphic to a plane and has constant negative curvature equal to $-1$ .
Among its classical representations are the hyperboloid model, the Klein model, the Poincaré half-plane model, and the Poincaré disk model.
The present work uses the Poincaré disk model.

\begin{definition} 
\label{def:hyperbolic}
The Poincaré disk model is the triplet $(\mathbb{D}, \mathsf{d}_{\mathbb{D}}, \mathsf{m}_{\mathbb{D}})$ constituted of the open unit disk $\mathbb{D} = \lbrace z \in \mathbb{C} : \lvert z \rvert < 1\rbrace$, the \emph{hyperbolic metric} defined by
    \begin{align}
    \label{eq:hyper-metric}
        \mathrm{d}\mathsf{m}_{\mathbb{D}}(z) = \frac{4 \mathrm{d}z}{\left(1 - \lvert z \rvert^2\right)^2}
    \end{align}
    where $\mathrm{d}z$ refers to the standard Lebesgue measure on $\mathbb{C}$, and the associated \emph{hyperbolic distance} defined for $z,w \in \mathbb{D}$ by
    \begin{align}
        \mathsf{d}_{\mathbb{D}}(z,w) = 2 \mathrm{tanh}^{-1}(\mathsf{p}_{\mathbb{D}}(z,w)), \quad \mathsf{p}_{\mathbb{D}}(z,w) = \frac{\lvert z - w \rvert}{\lvert 1 - \overline{w} z \rvert} \label{eq:dph_disk}
    \end{align}
with $\mathsf{p}_{\mathbb{D}}$ the \emph{pseudo-hyperbolic} distance.
\end{definition}

The isometries under the hyperbolic distance $\mathsf{d}_{\mathbb{D}}$ of the Poincaré disk are exactly the homographies of the form
\begin{align}
    z \mapsto \frac{az + b}{\overline{b}z + \overline{a}}, \quad a,b\in \mathbb{C}, \, \lvert a \rvert^2 - \lvert b \rvert^2 =1.
\end{align}

\begin{definition}
    Let $\alpha > 0 $ be a real number, the \emph{hyperbolic Gaussian analytic function} is defined on the Poincaré disk $\mathbb{D}$ by
    \begin{align}
        \forall z \in \mathbb{D}, \quad \GAF_{\mathbb{D}}^{(\alpha)}(z) = \sum_{n = 0}^\infty \wn_n \sqrt{\frac{\Gamma(\alpha + n)}{n!}}z^n
    \end{align}
    where $\left( \wn_n \right)_{n \in \mathbb{N}}$ is a sequence of independent and identically distributed standard Gaussian random variables, $\wn_n \sim \mathcal{N}_{\mathbb{C}}(0,1)$.
    By~\cite[Lemma 2.2.3]{Hough2009}, for $\alpha >0$, $\GAF_{\mathbb{D}}^{(\alpha)}$ is well-defined and almost-surely an analytic function on $\mathbb{D}$.
\end{definition}

Gaussian analytic functions have received much interest from the probability and spatial statistics communities, in particular due to the 
remarkable properties of their zeros~\cite{Hough2009}. 
Indeed, since Gaussian analytic functions are almost-surely analytic, with probability one, their zeros constitute a random set of isolated points in $\mathbb{C}$.
Further, \cite[Lemma 2.4.1]{Hough2009} ensures that the random zeros of a Gaussian analytic function are almost-surely \emph{simple}.
Hence, the zero set of a Gaussian analytic function constitutes a \emph{simple point process}~\cite{daley2003introduction,moller2003statistical}, defined as a random variable taking values in the configurations of points in $\mathbb{C}$.
Several Gaussian analytic functions, among which the hyperbolic Gaussian analytic function, are distinguished by the fact that their zeros are distributed very uniformly in their definition domain, which is mathematically formalized as an \emph{invariance} under isometries~\cite[Section 2.3]{Hough2009}.
This invariance property is the cornerstone of the zero-based signal processing procedures designed for detection~\cite{Bardenet2020,Pascal2022}, denoising~\cite{Koliander2019,Bardenet2020} and component separation~\cite{Flandrin2015}. The interested reader can refer to~\cite{pascal2024pointprocessesspatialstatistics} for a review of the use of spatial statistics in time--frequency analysis.

\begin{theorem}\cite[Proposition 2.3.4]{Hough2009}
For $\alpha >0$, the point process formed by the zeros of $\GAF_{\mathbb{D}}^{(\alpha)}$ is invariant under the isometries of the hyperbolic disk.
\label{thm:isometries}
\end{theorem}

\subsection{Link with the zeros of the hyperbolic Gaussian analytic function}

A major consequence of Theorem~\ref{thm:AST} is that, since $\vartheta$ is invertible and $\lambda$ never vanishes, the zeros of the Analytic Stockwell Transform of a finite energy signal are in one-to-one correspondance with the zeros of an analytic function $F$ on the unit disk.
As so, the zeros of the Analytic Stockwell Transform form a collection of isolated points in $\mathbb{D}$.
Further, in the presence of noise in the signal, $f$ is random and then the zeros of $S_{\varphi_{\beta}} f$ are random isolated points in $\mathbb{D}$, and thus form a point process.

In line with~\cite{Abreu2018,Bardenet2021,Pascal2022}, the purpose of this section is to study a remarkable instance of such point processes, that is  the zeros of the Analytic Stockwell Transform of \emph{white noise}, and  to demonstrate that it coincides in law with the zeros of the hyperbolic Gaussian analytic function introduced in Section~\ref{ssec:GAF-PP}.
Following~\cite{Abreu2018,Koliander2019,Bardenet2020,Bardenet2021}, this connection is established in two steps: first the continuous white noise adapted to the considered transform is rigorously constructed; second the equality, in law, of the point process of the zeros of the transform of white noise and of the zeros of the associated Gaussian analytic function is proven.
The following lemma will appear as an important building block for the design of the continuous white noise adapted to the Analytic Stockwell Transform.

\begin{lemma}
\label{lem:def_fn}
    Let $\beta > 0$ and consider the family of functions $\left\lbrace f_n \right\rbrace_{n \in \mathbb{N}}$ defined through their Fourier transforms such that, $\forall n \in \mathbb{N}$,
\begin{align}
\label{eq:def_fn}
    \forall \nu \in \mathbb{R}_+, \quad \widehat{f}_n(\nu) = \sqrt{\frac{n!}{\Gamma(2\beta + n +1)}} \, 2\pi\, \e^{- 2\pi \nu} (2\pi)^{2\beta} \nu^\beta L_n^{(2\beta)}(4\pi\nu),
\end{align}
where $L_n^{(2\beta)}$ denotes the \emph{Laguerre polynomial} of order $n$ with parameter $2\beta$, defined by the Rodrigues formula
    \begin{align}
    \label{eq:rodrigues}
        L_n^{(2\beta)}(\mu) =  \frac{\e^\mu \mu^{-2\beta}}{n!} \dfrac{\mathrm{d}^n}{\mathrm{d}\mu^n}\left[  \e^{-\mu} \mu^{2\beta + n}\right](\mu).
    \end{align}
There exists a smooth nonvanishing function $\eta_{\beta} : \mathbb{R} \times \mathbb{R}_+$ such that, for all $n \in \mathbb{N}$, $z=x+\ii \xi^{-1} \in \mathbb{C}$, the Analytic Stockwell Transform of $f_n$ has the closed-form expression:
\begin{align}
\begin{split}
\label{eq:Sbeta-fn}
    S_{\varphi_{\beta}} f_n(x,\xi) = \eta_\beta(x,\xi) \sqrt{\frac{\Gamma(2\beta + n + 1)}{  n!}} \, \vartheta(z)^n.
     \end{split}
\end{align}
\end{lemma}

\begin{remark}
The above lemma is an original contribution of the present work.
Indeed, first, the orthonormal basis used to connect the Cauchy Wavelet Transform of white noise to the hyperbolic Gaussian analytic function in~\cite{Abreu2018,Koliander2019,Bardenet2021} has to be adapted to the Analytic Stockwell Transform of Section~\ref{ssec:AST}.
Second, in~\cite{Abreu2018,Koliander2019} the orthonormal basis used is only defined implicitly.
Third, it seems that the explicit derivations performed in~\cite[Theorem~2.3]{Bardenet2021}, relying on~\cite[Equation~(2.15)]{Bardenet2021}, together with the computations in~\cite[Section 4]{daniel2008remarks}, contains erroneous prefactors that might impact the definition of white noise adapted to the Cauchy Wavelet Transform used to prove~\cite[Theorem~2.3]{Bardenet2021}. 
Hence, since Equation~\eqref{eq:Sbeta-fn} is the cornerstone of the demonstration of the connection between the hyperbolic Gaussian analytic function and the Analytic Stockwell Transform of white noise, for the sake of completeness, the proof below details the entire calculation, including the proposed corrected prefactors and normalizations.
\end{remark}

\begin{proof}
First, remark that for $f \in \Ell(\mathbb{R})$, the analytic factor of the Stockwell transform, corresponding to
    \begin{align*}
        \xi^{\beta + 1/2}W_{\psi_\beta} f (x, \xi^{-1}) &= \xi^{\beta+1} \int_{-\infty}^\infty f(t) \overline{\psi_\beta\left(\frac{t-x}{\xi^{-1}} \right)} \, \dd t,
    \end{align*}
    involves the convolution between $f$ and $h$ defined by $h(t) = \overline{\psi_{\beta}(- t /\xi^{-1})}$.
    Then, using that the Fourier transform of the convolution between two functions is the product of the Fourier transforms of these functions leads to
    \begin{align}
    \label{eq:Wbeta}
        \xi^{\beta + 1/2} W_{\psi_\beta} f (x, \xi^{-1}) &=  \xi^{\beta+1 } \int_{-\infty}^\infty \widehat{f}(\nu) \widehat{h}(\nu) \e^{2\pi \ii \nu x}\, \dd \nu.
    \end{align}
    Using the symmetries of the Fourier transform and the expression of the Fourier transform of the Cauchy wavelet provided in Equation~\eqref{eq:def-cauchy} yields
    \begin{align*}
        \widehat{h}(\nu) = \xi^{-1} \overline{\widehat{\psi}_{\beta}(\xi^{-1}\nu)} = 
        \left\lbrace\begin{array}{cl}
        \xi^{-1 - \beta} \nu^{\beta}  \e^{-2\pi \xi^{-1}\nu} & \text{if } \nu > 0 \\
        0 & \text{otherwise}
        \end{array}
        \right.
    \end{align*}
    which, when injected into Equation~\eqref{eq:Wbeta}, leads to
    \begin{align}
    \begin{split}
        \xi^{\beta + 1/2} W_{\psi_\beta} f (x, \xi^{-1}) &=  \xi^{\beta+1} \int_{-\infty}^\infty \widehat{f}(\nu) \widehat{h}(\nu) \e^{2\pi \ii \nu x}\, \dd t\\
        & = \int_{0}^\infty \widehat{f}(\nu) \nu^{\beta} \e^{-2\pi \xi^{-1}\nu}  \e^{2\pi \ii \nu x} \, \dd \nu\\
        & = \int_{0}^\infty \widehat{f}(\nu) \nu^{\beta}   \e^{2\pi \ii \nu (x+\ii\xi^{-1})} \, \dd \nu. \label{eq:Wpsif}
        \end{split}
    \end{align}

Then, injecting the expression of the Fourier transform of $f_n$ provided in Equation~\eqref{eq:def_fn} into the expression of the analytic wavelet transform of Equation~\eqref{eq:Wpsif} and setting $z = (x+\ii\xi^{-1})$, one obtains 
\begin{align}
\xi^{\beta + 1/2} W_{\psi_\beta} f_n (x, \xi^{-1}) 
    = \hspace{3.55mm}&  \int_{0}^\infty \widehat{f}_n(\nu) \nu^{\beta}   \e^{2\pi \ii \nu z} \, \dd \nu\\ 
     = \hspace{3.5mm} & \sqrt{\frac{n!}{\Gamma(\beta + n +1)}} \,  2\pi  \! \! \int_{0}^\infty \!\! \e^{2\pi \nu(\ii z -1)} (2\pi \nu)^{2\beta}  L_n^{(2\beta)}(4\pi\nu) \, \dd \nu \nonumber\\
     \overset{\mu \, = \, 2\pi\nu}{=} & \sqrt{\frac{n!}{\Gamma(\beta + n +1)}} \, \int_{0}^\infty \e^{ \mu(\ii z -1)} \mu^{2\beta}  L_n^{(2\beta)}(\mu) \, \dd \mu. \nonumber
     \label{eq:Wbeta-fn}
\end{align}
Replacing the expression of the Laguerre polynomial by its Rodrigues formula of Equation~\eqref{eq:rodrigues} one gets
\begin{align}
    \begin{split}
        \int_{0}^\infty  \e^{ \mu(\ii z -1)} \mu^{2\beta}  &L_n^{(2\beta)}(2\mu) \, \dd \mu \\
        = \hspace{3mm}& \int_{0}^\infty \e^{ \mu(\ii z -1)} \mu^{2\beta}  \frac{\e^{2\mu} (2\mu)^{-2\beta}}{n!} \dfrac{\mathrm{d}^n}{\mathrm{d}\mu^n}\left[  \e^{-\mu} \mu^{2\beta + n}\right](2\mu) \, \dd \mu \\
        = \hspace{3mm} & \frac{1}{ 2^{2\beta} n!} \int_{0}^\infty \e^{ \mu(\ii z +1)}    \dfrac{\mathrm{d}^n}{\mathrm{d}\mu^n}\left[  \e^{-\mu} \mu^{2\beta + n}\right](2\mu) \, \dd \mu\\
        \overset{\nu \, =  \, 2 \mu}{=} & \frac{1}{ 2^{2\beta+1} n!} \int_{0}^\infty \e^{ \nu\frac{\ii z +1}{2}}    \dfrac{\mathrm{d}^n}{\mathrm{d}\nu^n}\left[  \e^{-\nu} \nu^{2\beta + n}\right](\nu) \, \dd \nu\\
    \end{split}
\end{align}
which is the Laplace transform\footnote{Let $g : \mathbb{R}\rightarrow \mathbb{C}$. The Laplace transform of $g$, denoted $\mathcal{L}\{g\} : \mu \subseteq \mathbb{C} \rightarrow \mathbb{C}$, is a function of the complex variable $s$ defined by:  $\displaystyle
    \mathcal{L}\lbrace g\rbrace(s) = \int_{\mathbb{R}} \e^{-\nu s}g(\nu) \, \mathrm{d}\nu$, $s \in \mu$.} of the $n$th order derivative of $g_n(\nu) = \e^{-\nu} \nu^{2\beta+n}$ evaluated at $-(\ii z + 1)/2$, i.e.,
\begin{align}
\label{eq:integral_L2beta}
    \int_{0}^\infty  \e^{ \mu(\ii z -1)} \mu^{2\beta}  L_n^{(2\beta)}(2\mu) \, \dd \mu = \frac{1}{ 2^{2\beta+1} n!} \mathcal{L}\left\lbrace g_n^{(n)}\right\rbrace\left( -\frac{\ii z+1}{2}\right).
\end{align}
To get a closed-form expression of $\mathcal{L}\left\lbrace g_n^{(n)}\right\rbrace$, the following lemma is needed.

\begin{lemma}
    Let $n \in \mathbb{N}$, and $g :\mathbb{R} \rightarrow \mathbb{C}$ be $n$ times differentiable.
    Assume that the $n$th derivative of $g$, denoted $g^{(n)}$, is of exponential type\footnote{A function $g : \mathbb{R} \rightarrow \mathbb{C}$ is said to be of \emph{exponential type} if there exist $M > 0$ and $\tau  >0$ such that $\lvert g(t) \rvert \leq M \e^{t/\tau}$ in the limit $t \rightarrow \infty$}.
    Then, the Laplace transform of $g^{(n)}$ writes
    \begin{align}
\label{eq:prop_Laplace_n}
    \mathcal{L}\left\lbrace g^{(n)}\right\rbrace(s) = s^n \mathcal{L}[g](s) - \sum_{k = 0}^{n-1} s^{n- k -1} g^{(k)}(0^+)
\end{align}
where $g^{(k)}(0^+)$ denotes the limit of $g^{(k)}(t)$ as $t \rightarrow 0$, $t > 0$. 
\label{lem:Laplace-gn}
\end{lemma}
\begin{proof}
The proof derives from a recursion on the differentiation order $n$.
\end{proof}

Hence, the first step to compute the Laplace transform of $g_n^{(n)}$ is to derive an expression of the Laplace transform of $g_n$.
For all $s \in \mathbb{C}$ satisfying $\mathrm{Re}(1+s) > 0$, one has
\begin{align}
\begin{split}
\label{eq:Laplace_gn}
    \mathcal{L}\{g_n\}(s) = \int_0^\infty \e^{-s \nu} \e^{-\nu} \nu^{2\beta+n} \, \mathrm{d}\nu 
    & = \frac{\Gamma(2\beta + n + 1)}{(1+s)^{2\beta + n + 1}}
    \end{split}
\end{align}
leveraging the change of variable $\nu' = \nu (1+s)$ and the definition of the gamma function. The second step is to compute the limits $g^{(k)}(0^+)$ for $k \in \lbrace 0, \hdots, n-1 \rbrace$.
To that aim, a recursion shows that for all $k \in  \lbrace 0, 1, \hdots, n \rbrace$, there exists a polynomial $p_k$ of order $n$, such that
\begin{align}
\label{eq:gk-pk}
    g^{(k)}(\nu) = \e^{-\nu}\nu^{2\beta} p_k(\nu)
\end{align}
and thus, since $\beta > 0$, $\forall k \in  \lbrace 0, 1, \hdots, n-1 \rbrace, \, g^{(k)}(0^+) = 0$. 
Further, for $k = n$, Equation~\eqref{eq:gk-pk} shows that $g_n^{(n)}$ is of exponential type.
Altogether, Lemma~\ref{lem:Laplace-gn} applies and the combination of Equations~\eqref{eq:prop_Laplace_n} and \eqref{eq:Laplace_gn} enables to obtain a closed-form expression of the integral in Equation~\eqref{eq:integral_L2beta}
\begin{align}
\label{eq:wave-fn}
    \int_{0}^\infty  \e^{ \mu(\ii z -1)} \mu^{2\beta}  L_n^{(2\beta)}(2\mu) \, \dd \mu  \nonumber
    &= \frac{1}{ 2^{2\beta+1} n!} \left( -\frac{\ii z+1}{2}\right)^n \frac{\Gamma(2\beta + n + 1)}{(1- \frac{\ii z + 1}{2})^{2\beta + n + 1}}\\
    & = \frac{1}{  n!} \left( -\ii z - 1\right)^n \frac{\Gamma(2\beta + n + 1)}{(1- \ii z )^{2\beta + n + 1}}\\
    & = \frac{1}{  n!} \left( \frac{-\ii z-1}{1- \ii z }\right)^n \frac{\Gamma(2\beta + n + 1)}{(1- \ii z )^{2\beta  + 1}} \nonumber
    \\
    & = \left(\frac{\ii}{z + \ii}\right)^{2\beta  + 1} \frac{\Gamma(2\beta + n + 1)}{  n!} \left( \frac{z-\ii}{ z  + \ii}\right)^n.\nonumber
\end{align}
Equations~\eqref{eq:Wbeta-fn} and~\eqref{eq:wave-fn} show that
\begin{align}
    \xi^{\beta + 1/2} W_{\psi_\beta} f_n (x, \xi^{-1}) = \left(\frac{\ii}{z + \ii}\right)^{2\beta  + 1}   \sqrt{\frac{\Gamma(2\beta + n + 1)}{  n!}} \, \vartheta(z)^n
\end{align}
where $\vartheta$ is the Cayley transform defined in Equation~\eqref{eq:Cayley}.
Finally, remembering that $z = x + \ii \xi^{-1}$, let $\eta_{\beta} : \mathbb{R}\times \mathbb{R}_+ \rightarrow \mathbb{C}$ be defined by
\begin{align}
   \forall (x,\xi) \in \mathbb{R}\times \mathbb{R}_+, \quad \eta_{\beta}(x, \xi) = \sqrt{\xi} \e^{- 2\pi \ii \xi x} \left(\frac{\ii}{x + \ii \xi^{-1} + \ii}\right)^{2\beta  + 1}.
\end{align}
Then, for any $\beta > 0$, $\eta_{\beta}$ is a well-defined smooth nonvanishing function on the upper half-plane, and the Analytic Stockwell Transform of $f_n$ writes
\begin{align}
\begin{split}
    S_{\varphi_{\beta}} f_n(x,\xi) 
     &= \sqrt{\xi} \e^{- 2\pi \ii \xi x} W_{\psi_\beta} f_n(x, \xi^{-1}) \\&= \eta_\beta(x,\xi) \sqrt{\frac{\Gamma(2\beta + n + 1)}{  n!}} \, \vartheta(z)^n.
     \end{split}
\end{align}
\end{proof}

\begin{theorem}
\label{thm:AST-GAF}
    Let $\beta > 0$, $\psi_{\beta} \in L^2(\mathbb{R})$ be the Cauchy wavelet of parameter $\beta$ and $\varphi_{\beta}(t) = \psi_{\beta}(t) \e^{-2\pi \ii t}$.
    Then,  up to a conformal transform, the zeros of the Analytic Stockwell Transform with analysis window $\varphi_{\beta}$ of white noise coincide in law with the zero set of the hyperbolic Gaussian analytic function of parameter $2\beta + 1$.
    Consequently, the zero set of the Analytic Stockwell Transform of white noise is invariant under isometries of $\mathbb{D}$.
\end{theorem}

\begin{remark}
Establishing a connection between time--frequency transforms and Gaussian analytic functions with isometry-invariant zero sets requires to consider the \emph{complex-valued} white noise~\cite[Section 4]{Bardenet2020}, \cite{Bardenet2021,Abreu2018,Koliander2019,Pascal2022}.
Thus, all along this work, ``white noise'' will refer to a complex random variable.\footnote{The interested reader can refer to~\cite[Section 3]{Bardenet2020} for a study of the far more intricate case of real white noise, which will not be discussed in the present work.}
\end{remark}

\begin{proof}
This proof follows closely the reasoning developed in~\cite[Theorem 4.3]{Bardenet2021}, linking the Cauchy Wavelet Transform of white noise to the hyperbolic Gaussian analytic function.
It is sketched here for the sake of completeness.

The first step consists in defining rigorously the continuous white noise adapted to the considered transform.
To that aim, two main paths have been proposed, respectively described in~\cite{Koliander2019,Bardenet2020} and in~\cite{Bardenet2021}, the later being used in the following.
Considering $\mathcal{H} = \Ell(\mathbb{R})$ and the orthonormal basis $\lbrace f_n \rbrace_{n\in \mathbb{N}}$ introduced in Lemma~\ref{lem:def_fn}, define a novel norm $\lVert \cdot \rVert_{\Theta}$ such that
\begin{align}
    \forall f \in \mathcal{H}, \quad \lVert f \rVert_{\Theta}^2 = \sum_{n \in \mathbb{N}} \frac{1}{1+n^2} \lvert \langle f, f_n \rangle\rvert^2.
\end{align}
Then, the completion $\Theta$ of $\Ell(\mathbb{R})$ for the norm $\lVert \cdot \rVert_{\Theta}$ is such that the series
\begin{align}
\label{eq:def_wn}
    \wn = \sum_{n \in \mathbb{N}} \wn_n f_n, \quad \wn_n \sim \mathcal{N}_{\mathbb{C}}(0,1)
\end{align}
with $\left( \wn_n \right)_{n \in \mathbb{N}}$ a sequence of independent and identically distributed standard Gaussian random variables has a well-defined limit.
Furthermore, by orthonormality of the basis $\lbrace f_n \rbrace_{n\in \mathbb{N}}$ \cite[Proposition 3.3]{Bardenet2021} ensures that Equation~\eqref{eq:def_wn} defines a Gaussian random variable with the characteristic function of white noise.
Note that, as emphasized in~\cite[Section 3]{Bardenet2021} $\Theta$ depends on the orthonormal basis $\lbrace f_n\rbrace$, hence the importance of Lemma~\ref{lem:def_fn} deriving a basis suited to the Analytic Stockwell Transform.

The second step is to show that the Analytic Stockwell Transforms extends to $\Theta$, which then ensures that the transform of the continuous white noise is well-defined. 
To that aim, for $\beta > 0$ fixed, let $\lbrace \Psi_n \rbrace_{n \mathbb{N}}$ be the sequence of analytic functions on $\mathbb{D}$ defined by
\begin{align}
    \Psi_n(w) = \sqrt{\frac{\Gamma(2\beta + n + 1)}{n!}}w^n.
\end{align}
For $\mathfrak{K}\subset \mathbb{D}$ a compact, let $0 < R < 1$ be such that $\mathfrak{K} \subset \overline{\mathsf{B}}(0,R)$, then
\begin{align}
    \begin{split}
    \sum_{n\in \mathbb{N}} (1+n^2) \lvert \Psi_n(w)\rvert^2 &= \sum_{n\in \mathbb{N}} (1+n^2) \frac{\Gamma(2\beta + n + 1)}{n!} \lvert w\rvert^{2n} \\
    &\leq \sum_{n\in \mathbb{N}} (1+n^2) \frac{\Gamma(2\beta + n + 1)}{n!} R^{2n}\\
    & < \infty \quad \text{since } \lvert R  \rvert < 1.
    \end{split}
\end{align}
Hence, the ($C_{\mathfrak{K}}$ revisited) condition of \cite[Theorem 3.4]{Bardenet2021} is satisfied and
\begin{align}
\label{eq:def_L}
\mathcal{T} f(w) = \sum_{n \in \mathbb{N}} \langle f, f_n \rangle \Psi_n(w)
\end{align}
yields a well-defined and continuous transform $\mathcal{T} : \Ell(\mathbb{R}) \rightarrow \mathcal{A}(\mathbb{D})$.
Then, Lemma~\ref{lem:def_fn} shows that for any $n \in \mathbb{N}$,
\begin{align}
\begin{split}
    \forall (x,\xi) \in \mathbb{R}\times \mathbb{R}_+, \quad S_{\varphi_\beta}f_n(x, \xi) &= \eta_{\beta}(x,\xi) \Psi_n(\vartheta(z)) \\
    &= 
    \eta_{\beta}(x,\xi)\mathcal{T}f_n(\vartheta(z)), \quad z = x + \ii \xi^{-1}
    \end{split}
\end{align}
which extends to any $f \in \Ell(\mathbb{R})$ by continuity of $\mathcal{T}$.
Applying \cite[Theorem 3.4]{Bardenet2021} proves that $\mathcal{T}$ extends continuously to $\Theta$. 
Consequently, $\mathcal{T}\wn$ is well-defined and it follows from the fact that $\eta_{\beta}$ is nonvanishing that the relation
\begin{align}
\label{eq:def_Sbeta-wn}
    \forall (x,\xi) \in \mathbb{R}\times \mathbb{R}_+, \quad S_{\varphi_\beta}\wn(x, \xi)  = \eta_{\beta}(x,\xi) \mathcal{T}\wn(\vartheta(z)), \quad z = x + \ii \xi^{-1}
\end{align}
yields a rigorous definition to the Analytic Stockwell transform of white noise.
Finally, by the definition of $\wn$ and of $\mathcal{T}$ provided in Equations~\eqref{eq:def_wn} and~\eqref{eq:def_L},
\begin{align*}
   \forall w \in \mathbb{D}, \, \, \,  \mathcal{T}\wn(w) &= \sum_{n \in \mathbb{N}} \wn_n \Psi_n(w) =\sum_{n \in \mathbb{N}} \wn_n \sqrt{\frac{\Gamma(2\beta + n + 1)}{n!}}w^n 
   = \GAF^{(2\beta+1)}(w)
\end{align*}
from which one deduces that for all $(x,\xi) \in \mathbb{R}\times \mathbb{R}_+$, and for  $z = x + \ii \xi^{-1}$
\begin{align}
\label{eq:Sbeta-GAF}
S_{\varphi_\beta}\wn(x, \xi)  = \eta_{\beta}(x,\xi) \GAF^{(2\beta+1)}(\vartheta(z)).
\end{align}
Since $\eta_{\beta}$ is nonvanishing and $\vartheta$ is bijective, Equation~\eqref{eq:Sbeta-GAF} ensures that the zeros of the Analytic Stockwell Transform of white noise are in one-to-one correspondance with the zeros of the hyperbolic Gaussian analytic function of parameter $2\beta +1$.
\end{proof}

\section{Numerical Experiments and Spatial Statistics of Zeros}
\label{sec:numeric}

The purpose of the present section is twofold.
First, illustrations of the zero set of the Analytic Stockwell Transform of white noise are provided, enabling the reader to observe qualitatively the \emph{hyperbolic} uniform spread of zeros in the Poincaré disk.
Second, the characterization of the distribution of zeros, established in Theorem~\ref{thm:AST-GAF}, is supported by quantitative numerical evidence obtained by thorough comparisons of the theoretical and empirical \emph{spatial statistics} of the point process of zeros.
A documented Python toolbox implementing the Analytic Stockwell Transform and performing the spatial statistics analysis of the zeros of the transform, thus enabling to reproduce all the experiments, plots and figures presented in Section~\ref{sec:numeric} has been made publicly available by the authors.\footnote{\label{ft:toolbox}\url{https://github.com/courbot/ast}}

\subsection{Discrete Time--Frequency Analysis}

In Section~\ref{sec:GST}, the Analytic Stockwell Transform has been defined for \emph{continuous} signals $f \in \Ell(\mathbb{R})$.
Though, in practice, one only measures a finite number $N \in \mathbb{N}^*$ of values of the signal of interest in a bounded time window
\begin{align}
\label{eq:xn}
   \forall n \in \lbrace 1,\hdots, N\rbrace, \quad  y_n = f\left(x_n \right), \quad x_n = x_{\min} + n \Delta_{x},
\end{align}
corresponding to \emph{sampling} $f$ in  $[x_{\min}, x_{\max}]$ at equally spaced time $x_1, \hdots, x_N$ with resolution $\Delta_{x} = (x_{\max} - x_{\min})/(N-1)$. 
This yields a \emph{discrete} signal $\mathbf{y}=(y_1, \dots, y_N)  \in \mathbb{C}^N$.
The continuous signal $f$ is then said to be sampled at frequency $1/\Delta_{x}$.
Furthermore, in practice, the frequency variable needs to be discretized as well. 
In the present work, motivated by the connection with the  Wavelet Transform, the discrete Analytic Stockwell Transform is computed on $M$ logarithmically spaced frequency channels: $\forall m \in \lbrace 1, \hdots, M\rbrace$,
\begin{align}
\label{eq:num}
   \log_2(\xi_m) = \log_2(\xi_{\min}) + m \Delta_{\log_2\xi}, \quad  \Delta_{\log_2\xi} = \frac{\log_2(\xi_{\max})-\log_2(\xi_{\min})}{M-1}
\end{align}
where $\Delta_{\log_2\xi}$ is the log-frequency resolution, and the range $[\xi_{\min}, \xi_{\max}]$ depends on the time--frequency pattern of interest~\cite[Theorem 5.15]{vetterli2014foundations}.

\begin{definition}
\label{def:discrete-AST}
Let $M,N \in \mathbb{N}^*$ and $\beta > 0$.
The \emph{Discrete Analytic Stockwell Transform} of parameter $\beta$ with $N$ time samples and $M$ frequency channels is the application $\mathbf{S}_{\varphi_\beta} : \mathbb{C}^N \rightarrow \mathbb{C}^{N\times M}$ associating to the discrete signal $\mathbf{y}=(y_1, \dots, y_N)  \in \mathbb{C}^N$, the matrix defined as
\begin{equation}
    (\mathbf{S}_{\varphi_\beta} \mathbf{y})_{j, m}=\sum_{n=0}^{N-1} y_n \overline{\varphi_\beta(\xi_m (x_n-x_j))} \e^{-\ii \frac{2 \pi}{N} \xi_m \cdot x_n}
    \label{eq:dicrete_ast}
\end{equation}
for indices $j \in\{0, \ldots, N-1\}, m \in\{0, \ldots, M-1\}$ and discrete times and frequencies $\lbrace x_n, \xi_m \rbrace_{n,m}$ defined in Equations~\eqref{eq:xn} and~\eqref{eq:num} respectively.
\end{definition}

\begin{definition}
Let $N \in \mathbb{N}^*$. The \emph{discrete} $N$-dimensional white noise $\boldsymbol{\wn} \in \mathbb{C}^N$ is the Gaussian random vector of zero mean and identity covariance matrix:
\begin{align}
    \boldsymbol{\zeta} = (\wn_1, \hdots, \wn_N), \quad \wn_n \sim \mathcal{N}(0,1)
\end{align}
where $\wn_1, \hdots, \wn_N$ are independent Gaussian variables.
\end{definition}

\begin{remark}
Following the same reasoning as the one developed in~\cite[Section 5]{Bardenet2021}, the Analytic Stockwell Transform of continuous white noise $S_{\varphi_{\beta}}\wn$ of Equation~\eqref{eq:def_Sbeta-wn} can be discretized consistently applying the discrete Stockwell transform of Definition~\ref{def:discrete-AST} to discrete white noise. This is notably due to orthogonality of the family $\lbrace f_n \rbrace_{n\in \mathbb{N}}$, defined in Lemma~\ref{lem:def_fn}, used to construct the continuous white noise $\wn$.
\end{remark}

Numerically, the discrete Analytic Stockwell Transform is computed in the Fourier domain, using a Riemann approximation of Equation~\eqref{eq:Wbeta} combined with Equation~\eqref{eq:AST-to-CWT}.
Following~\cite{Flandrin2015} the localization of zeros is performed using the \emph{Minimal Grid Neighbors} algorithm, introduced initially in the Euclidean setting in the code\footnote{https://perso.ens-lyon.fr/patrick.flandrin/zeros.html} accompanying the seminal paper~\cite{Flandrin2015}, and which were then used in Euclidean~\cite{Bardenet2020,Abreu2018}, hyperbolic~\cite{Koliander2019} and spherical geometries~\cite{Bardenet2021,Pascal2022}.
The method is described in details in~\cite[Section~4.3]{pascal2024pointprocessesspatialstatistics} in the context of standard time--frequency analysis, and is applied as is to the discrete Analytic Stockwell Transform of Definition~\ref{def:discrete-AST}.

\begin{remark}
Recently an \emph{Adaptive} Minimal Grid Neighbors algorithm has been derived in Euclidean geometry, benefiting from a solid theoretical ground and strong robustness~\cite{escudero2024efficient}.
Intensive numerical experiments has shown that both the \emph{Adaptive} and standard Minimal Grid Neighbors algorithms perform equivalently and significantly outperform thresholding strategies~\cite[Section 2.3]{escudero2024efficient}. Extending this theoretical analysis to the hyperbolic geometry is beyond the scope of the present work, though it constitutes a very interesting line of research for future work.
\end{remark}

Figure~\ref{fig:ast} displays the modulus of the discrete Analytic Stockwell Transform of parameter $\beta$ such that $\alpha = 2\beta + 1 = 300$ computed over a realization of discrete white noise of $N=4000$ points with sampling frequency $\nu_s = 400$~Hz, corresponding to a time resolution $\Delta_{x} = 1/\nu_s$, and for $M=600$ frequency channels ranging from $\xi_{\min} = 2^{-6}$ to $\xi_{\max} = 2^{3.3}$, together with the associated zeros.
The representation on the Poincaré disk, on the right of Figure~\ref{fig:ast}, illustrates that, in the \emph{hyperbolic} geometry in which the metric explodes as one gets closer to the border of the disk $\mathbb{D}$, the zeros spread very evenly.
Similarly, the representation in the time--scale plane, on the left of Figure~\ref{fig:ast}, shows the uniform spread in the hyperbolic upper-half plane.

\begin{figure}[ht]
    \centering
   \subfloat{\includegraphics[width=0.48\linewidth]{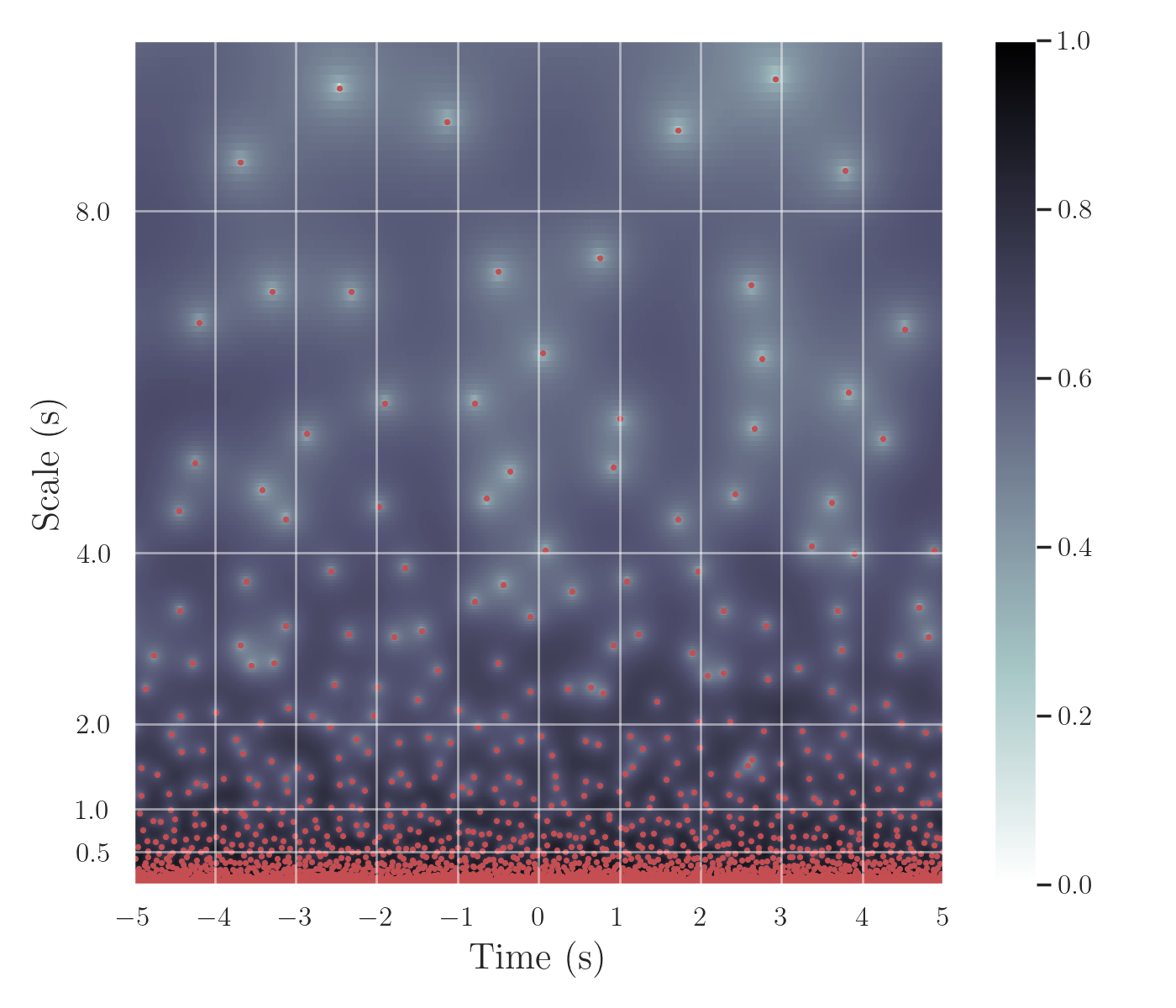}} \hfill
     \subfloat{\includegraphics[width=0.48\linewidth]{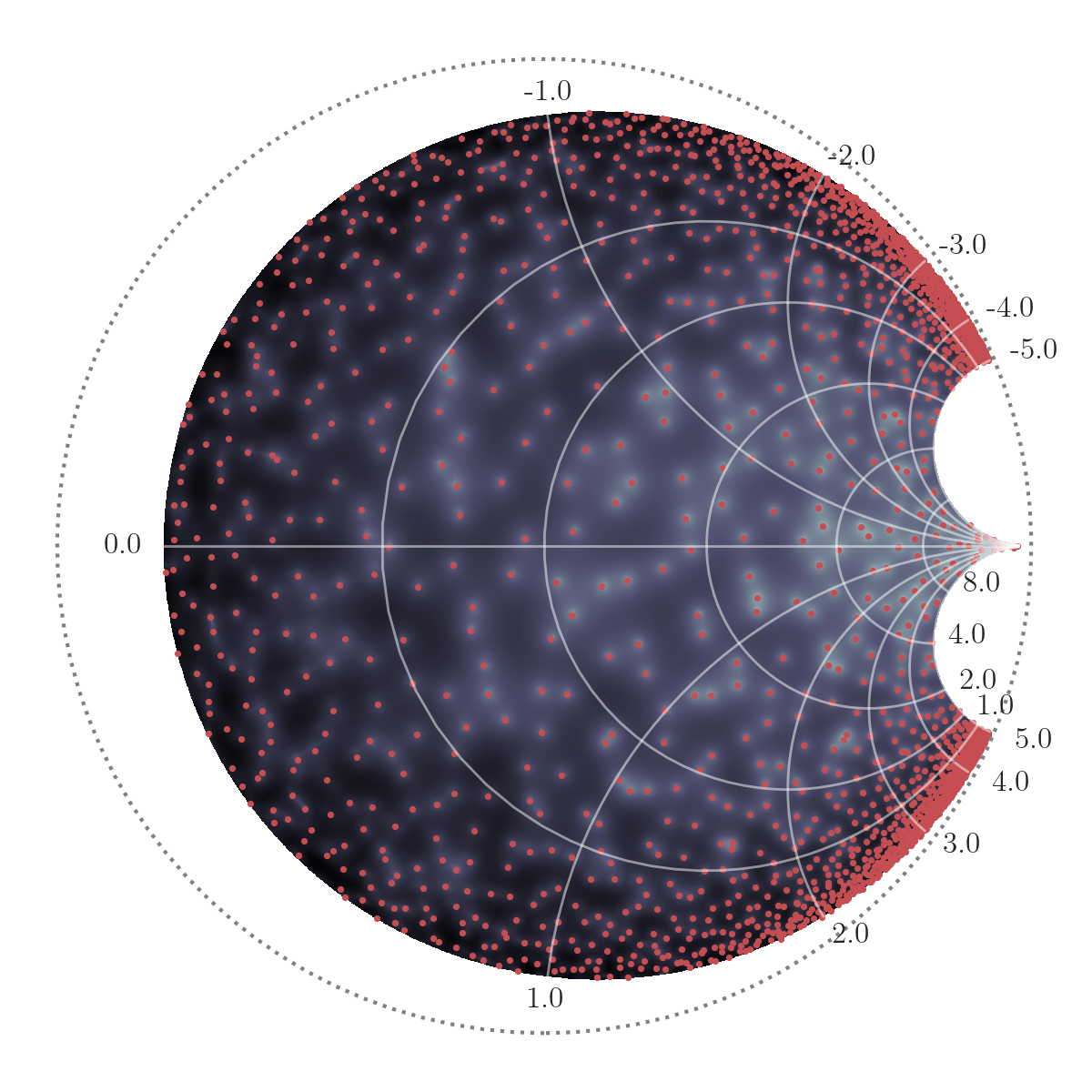}}  
     \caption{Log-modulus of the Analytic Stockwell Transform of white noise (background colormap) and its zero set (red dots) $\mathfrak{Z}^{(\alpha)}$. The discrete white noise is considered to span a time period of $x_{\max} - x_{\min} = 10$~s with a sampling frequency of $400$~Hz. The parameter of the Cauchy wavelet $\beta$ is chosen such that $\alpha=2\beta+1=300$.
     (Left) Representation in the $(x, \xi^{-1})$-plane where the inverse frequency $\xi^{-1}$ is called the \emph{scale} in reference to the related Cauchy Wavelet Transform.
      (Right) Representation in the Poincaré disk $\mathbb{D}$ parameterized by $ \vartheta(x + \ii \xi^{-1})$ for $(x,\xi^{-1})$ running over the same grid, where $\vartheta$ is the Cayley transform of Equation~\eqref{eq:Cayley}. The dotted circle represents the border of $\mathbb{D}$.
     \label{fig:ast}}
\end{figure}

\subsection{Spatial Statistics}

Point processes in a \emph{general} metric space can be characterized by their $k$-\emph{points correlation functions}~\cite[Section 5.4]{daley2003introduction}, describing the probabilistic interactions between the points of a realization of the point process, e.g., the short-range repulsion observed in the zero pattern of analytic time--frequency transforms~\cite{Flandrin2015,Koliander2019,Bardenet2020,Bardenet2021, Pascal2022}.
As emphasized in~\cite[Appendices B,C]{moller2003statistical}, spatial statistics tools can be derived from a general metric on any ambient space~\cite{daley2003introduction}. As this work focuses on the zeros of the Stockwell Transform lying in the Poincaré disk equipped with the hyperbolic distance introduced in Definition~\ref{def:hyperbolic}, all tools and mathematical results in the remaining of the paper are formulated directly in the hyperbolic geometry framework.

\begin{definition}
Let $\mathfrak{Z}$ be a point process on $\Lambda \subseteq \mathbb{C}$.
For $k \in \mathbb{N}^*$, if existing, the $k$-point correlation function $\rho_k$ of $\mathfrak{Z}$ is defined by 
\begin{align}
    \int_{\Lambda^k} \Psi(w_1, \hdots, w_k) \rho_k(w_1, \hdots, w_k)\, \mathrm{d}w_1\hdots\mathrm{d}w_k = \mathbb{E}\left[ \!  \sum_{\begin{array}{c} {\scriptstyle (w_1, \hdots, w_k) \in \mathfrak{Z}^k} \\ {\scriptstyle w_1 \neq \hdots \neq w_k} \end{array}} \! \! \! \! \! \! \ \ \! \! \! \! \ \! \! \! \!\Psi(w_1, \hdots, w_k)\right],
\end{align}
for any bounded compactly supported measurable map $\Psi : \Lambda^k \rightarrow \mathbb{C}$, and where $\mathrm{d}w$ denotes the metric in the ambient space, e.g., the hyperbolic metric $\mathrm{d}\mathsf{m}_{\mathbb{D}}$ in the Poincaré disk.
In particular, $\rho_1$ is often called the \emph{intensity} of the point process as, by definition, integrating it over a domain yields the expected number of points falling into.
As for the two-point correlation function $\rho_2$, it is commonly normalized to get the \emph{pair-correlation function}~\cite[Chapter 4]{moller2003statistical}, defined as
\begin{align}
    \widetilde{g}(w_1,w_2)= \frac{\rho_2(w_1, w_2)}{\rho_1(w_1)\rho_1(w_2)}, \quad w_1, w_2 \in \Lambda.
\end{align}
\end{definition}

By Proposition~\ref{thm:isometries}, the point process corresponding to the zero set of the hyperbolic Gaussian analytic function $\mathfrak{Z} = {\GAF_{\mathbb{D}}^{(\alpha)}}(\lbrace 0 \rbrace)$, and hence to the zeros of the Analytic Stockwell Transform of parameter $\beta = (\alpha - 1)/2$ of white noise, is invariant under isometries of $\mathbb{D}$, which has two major consequences.
First, the intensity $\rho_1$, defined relatively to the hyperbolic metric $\mathsf{m}_{\mathbb{D}}$ is constant.
Second, the two-point correlation function only depends on the hyperbolic distance between $w_1$ and $w_2$, or equivalently on the pseudo-hyperbolic distance.
Consequently, the pair correlation function of the zeros of the hyperbolic Gaussian analytic function of parameter $\alpha$ writes
\begin{align}
    g^{(\alpha)}(r) = \widetilde{g}^{(\alpha)}(w_1, w_2), \quad r = \mathsf{p}_{\mathbb{D}}(w_1 , w_2).
\end{align}
Remarkably, the joint intensities of zero sets of Gaussian analytic functions have closed-form expressions.
Applying the Edelman-Kostlan formula~\cite[Section 2.4.1]{Hough2009}~\cite[Proposition 16]{pascal2024pointprocessesspatialstatistics} yields the explicit expression\footnote{The difference from the formula derived in~\cite{Abreu2018,Koliander2019} comes from a difference in the normalization of the hyperbolic metric of Equation~\eqref{eq:hyper-metric} and of the metric used in~\cite{Abreu2018,Koliander2019}.} of the intensity of the zeros of $\GAF_{\mathbb{D}}^{(\alpha)}$:
\begin{align}
\label{eq:rho1}
    \rho_1^{(\alpha)} = \frac{\alpha}{\pi}.
\end{align}
To access the higher-order joint intensities, one can leverage the formulas
involving determinants and permanents provided in~\cite[Corrolary 3.4.1]{Hough2009}, to get the following explicit formula for the pair correlation functions~\cite{Abreu2018}
\begin{equation}
    g^{(\alpha)}(r) = \frac{ s^\alpha \big(\alpha(1-s) - s(1- s^\alpha)\big)^2 + \big(\alpha s^\alpha(1-s) - (1-s^\alpha)\big)^2 }{(1-s^\alpha)^3}.
    \label{eq:g_th}
\end{equation}
where $s = 1-r^2$.
By Theorem~\eqref{thm:AST-GAF}, Equation~\eqref{eq:g_th} corresponds both to the theoretical pair correlation function of the zeros of $\GAF^{(\alpha)}_{\mathbb{D}}$ and of the zeros of the Analytic Stockwell Transform of parameter $\beta = (\alpha-1)/2$ of white noise. For reference, $g^{(\alpha)}$ is displayed as a dashed black line in Figure~\ref{fig:numerical_g} for three values of $\alpha$.
One remarks, going from the left to the right of Figure~\ref{fig:numerical_g}, that as $\alpha$ increases the short-range repulsion between zeros visible at the vicinity of $r = 0$ weakens while the  range of attraction between zeros, correspond to the bump above one, gets narrower and shifts toward $r=0$.
To provide quantitative numerical evidence showing practical relevance of Theorem~\ref{thm:AST-GAF}, the theoretical expressions of both the first intensity  and the pair correlation function are to be compared to the empirical spatial statistics of the zeros of the Analytic Stockwell Transform of parameter $\beta = (\alpha-1)/2$ of white noise.

The need for accurate and robust estimators of empirical first intensity and pair correlation functions resulted in a rich variety of spatial statistics estimation strategies with solid theoretical guarantees~\cite[Section 4.3]{moller2003statistical}.
Notably, edge correction techniques have been developed to handle the fact that in practice only a bounded window is observed~\cite[Section 4.3.3]{moller2003statistical}.

\begin{figure}[ht]
 \centering  
 \subfloat[\label{sfig:isocontours}Isocontours (light orange to black circles) depicting equal pseudo-hyperbolic distances $\mathsf{p}_{\mathbb{D}}$, and thus by Equation~\eqref{eq:dph_disk} equal hyperbolic distance, to three selected centers (orange stars).]{\includegraphics[width=0.48\linewidth]{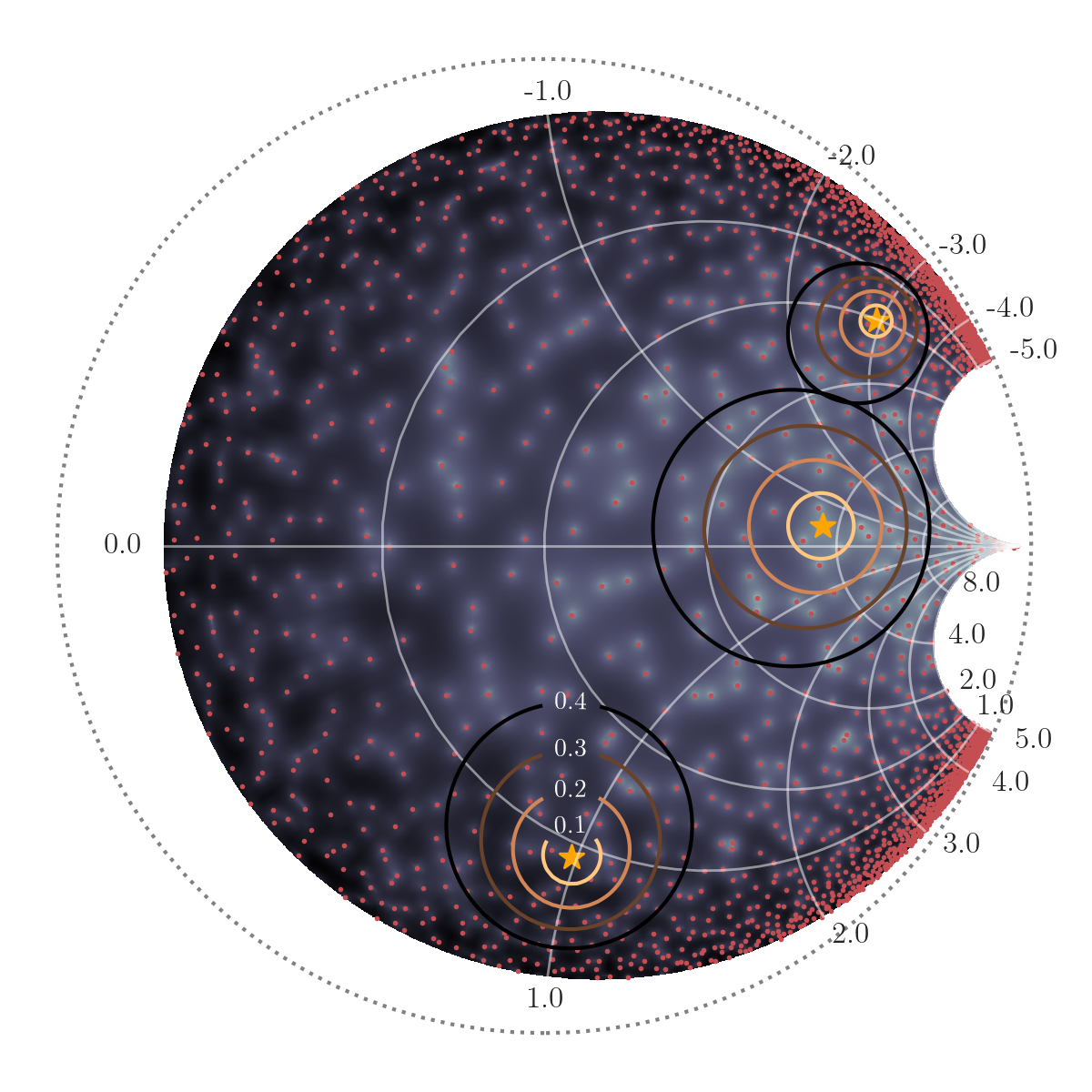}} \hfill
   \subfloat[\label{sfig:outer}\emph{Inner} points (red dots) in $\mathfrak{Z}^{(\alpha)} \setminus \mathfrak{B}^{(\alpha)}$ are used as center point to enumerate the density of zeros in their vicinity,  \emph{outer} points (pink dots) in $\mathfrak{B}^{(\alpha)}$ are too close to the border and hence excluded from the ensemble average.]{\includegraphics[width=0.48\linewidth]{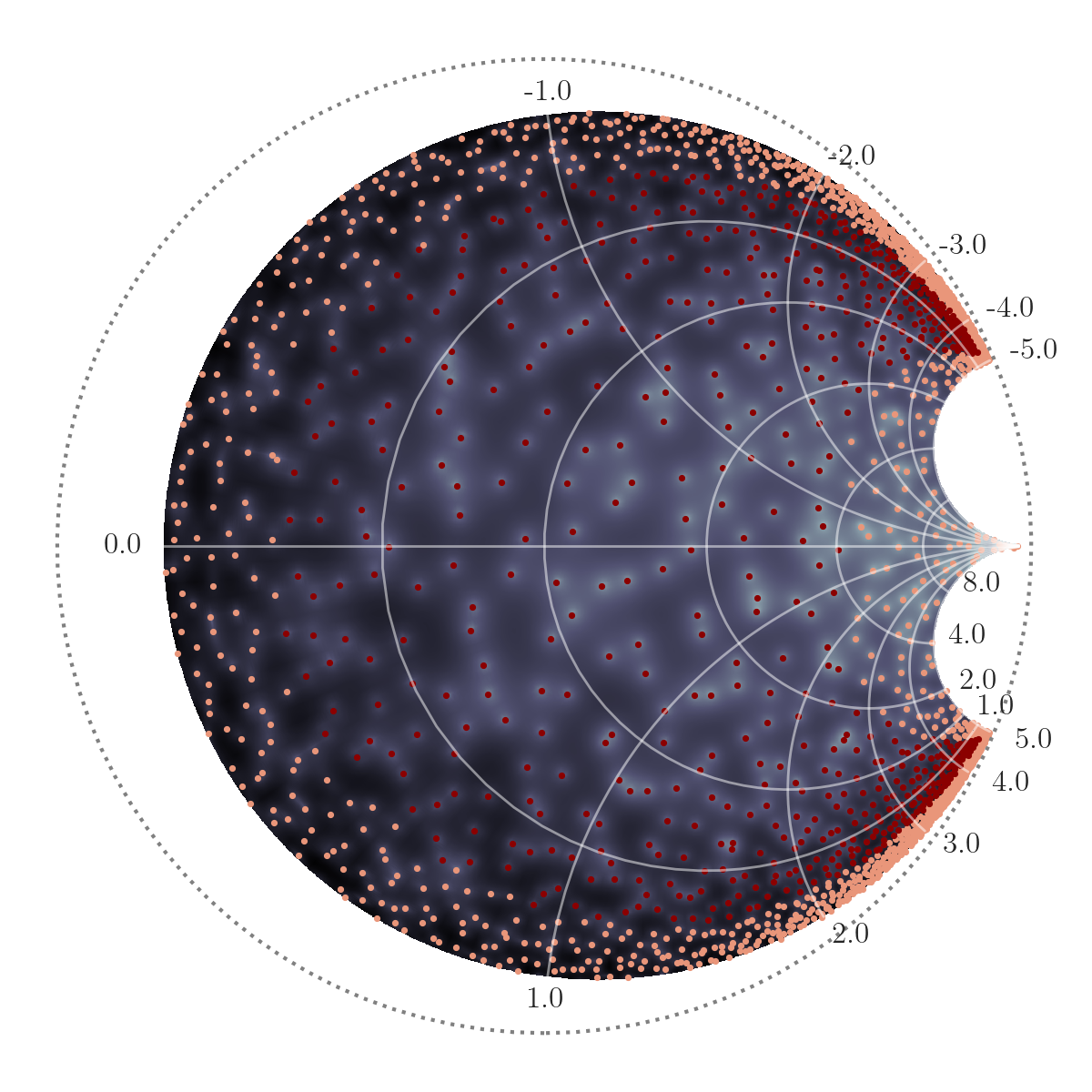}} 
     \caption{Key features of the implementation of the pair correlation function estimator  $\widehat{g}^{(\alpha)}$: counting zeros in hyperbolic disks (left) and border effect correction (right).}
     \label{fig:pictures_g}
 \end{figure}
 
First of all, to assess the ability of the hyperbolic counterpart of Minimal Grid Neighbor algorithm to correctly detect zeros in hyperbolic geometry with the discretization strategy described in Definition~\ref{def:discrete-AST}, the first intensity $\rho_1^{(\alpha)}$ of the zeros of the \emph{discrete} Stockwell Transform  of white noise of parameter $\beta = (\alpha-1)/2$ is estimated by counting the number of zeros in a growing hyperbolic disk centered at zero and of increasing hyperbolic radius $r'$ and dividing it by the \emph{hyperbolic} area $\mathcal{A}(r')$ of this disk, obtained integrating the hyperbolic metric~\eqref{eq:hyper-metric},
\begin{align}
\label{eq:hatrho1}
   \hat{\rho}_1^{(\alpha)} =\frac{1}{\mathcal{A}(r')} \sum_{z \in \mathfrak{Z}^{(\alpha)}} \mathds{1}_{\mathsf{d}_{\mathbb{D}}(z,0) < r'}, \quad  \mathcal{A}(r') = 4\pi \sinh(r'/2)^2
\end{align}
where $\mathfrak{Z}^{(\alpha)}$ denotes the zero set of the Analytic Stockwell transform of parameter $\beta = (\alpha-1)/2$ of white noise\footnote{Since hyperbolic and pseudo-hyperbolic distances are in one-to-one correspondence the above computations could be implemented equivalently replacing $\mathsf{d}_{\mathbb{D}}$ by $\mathsf{p}_{\mathbb{D}}$ and $r'$ by $r = \tanh(r'/2)$ according to~\eqref{eq:dph_disk}.} and $\mathds{1}_{\mathcal{A}}$ denotes the indicative function of the probabilistic event $\mathcal{A}$.

Then, the pair correlation function is estimated using the nonparametric estimator of~\cite[Section 4.3.5]{moller2003statistical}, the kernel being chosen as a step function of support size $h > 0$.
The number of zeros in the observed window being large enough, border effects are managed using the \emph{reduced-sampled} estimate~\cite[Section 4.3.3]{moller2003statistical}.
The pair correlation function estimator writes
\begin{equation}
	\label{eq:hatgalpha}
	\widehat{g}^{(\alpha)}(r) = \frac{(1-r^2)^2}{4\alpha h r} \sum_{z \in \mathfrak{Z}^{(\alpha)} \setminus \mathfrak{B}^{(\alpha)}}\sum_{w \in \mathfrak{Z}^{(\alpha)}} \mathds{1}_{\left\lbrace \left|\mathsf{p}_{\mathbb{D}}(z,w) - r\right| < \frac{h}{2} \right\rbrace}
\end{equation}
where $\mathfrak{Z}^{(\alpha)}$ is the zero set of the Analytic Stockwell transform of parameter $\beta = (\alpha-1)/2$ of white noise.
To correct the bias induced by the observation of a \emph{bounded} window the averaging over $z$ is performed only on the zeros which are far enough from the border of the observed hyperbolic window.
Denoting $ \mathfrak{B}^{(\alpha)} \subset \mathfrak{Z}^{(\alpha)}$ the zeros which are too close to the window's border, represented as pink dots in Figure~\ref{sfig:isocontours}, the averaging is performed on \emph{inner} points $z \in \mathfrak{Z}^{(\alpha)} \setminus  \mathfrak{B}^{(\alpha)}$, represented as dark red points in Figure~\ref{sfig:outer}, while the counting runs on every points $w \in \mathfrak{Z}^{(\alpha)}$.
The practical implementation of $\widehat{g}^{(\alpha)}$ amounts to building concentric rings of pseudo-hyperbolic radius $r$ around each point in the zeros set, as illustrated in Figure~\ref{sfig:isocontours}, counting the number of points in each pseudo-hyperbolic ring, which by~\eqref{eq:dph_disk} are also hyperbolic rinds, averaging it over the observed zeros, and finally normalizing by the theoretical first-intensity function provided in Equation~\eqref{eq:rho1}.

\begin{remark}
It is worth insisting on the fact that, although mostly used in Euclidean spaces, the aforementioned spatial statistics framework applies in any general metric space~\cite[Appendices B,C]{moller2003statistical}.
The customization to hyperbolic geometry then reduces to replacing the Euclidean metric and distance by their hyperbolic counterparts. This affects the shapes of the disks and the identification of \emph{outer points}, as can be seen on Figure~\ref{fig:pictures_g}, but not the theoretical results, proven in a general setting~\cite[Appendices B,C]{moller2003statistical}, \cite[Chapters 6,7]{daley2003introduction}. Notably, the edge correction implemented in Equation~\eqref{eq:hatgalpha} ensures the unbiasedness of the pair correlation function estimator regardless of the geometry.
\end{remark}

To compare the expected and empirical behaviors of the zeros of the Analytic Stockwell Transform of white noise, $R = 100$ realizations of the discrete white noise with $N=4000$ time stamps are generated. Their discrete Stockwell Transforms are computed at time resolution $ 1/\nu_s$ with $\nu_s = 4000$~Hz and across $600$ frequency channels. The Minimal Grid Neighbors algorithm then yields $R$ realizations of the point process $\mathfrak{Z}^{(\alpha)}$ each consisting in the zeros of a discrete Stockwell Transform of white noise. Finally the $R$ empirical first intensities $\widehat{\rho}_1^{(\alpha)}$ and pair correlation functions $\widehat{g}^{(\alpha)}$ are estimated from Equation~\eqref{eq:hatgalpha}, and averaged leading to the averaged empirical first intensity $\langle \widehat{\rho}_1^{(\alpha)}\rangle$ and  pair correlation function $\langle \widehat{g}^{(\alpha)}\rangle$ displayed respectively in Figure~\ref{fig:numerical_g} left and right columns as solid blue lines with star ticks, accompanied with their $5\%$ and $95\%$ quantiles.

Provided that the hyperbolic disk used to estimate the first intensity is large enough, the empirical estimates $\widehat{\rho}_1^{(\alpha)}$ are very close to the expected value provided in Equation~\eqref{eq:rho1}. This means that the Minimal Grid Neighbors algorithms, applied to the discrete Stockwell Transform of Definition~\ref{def:discrete-AST}, detects accurately enough the zeros of the transform so that the first order statistics are in very good agreement with the known statistics of the zeros of the hyperbolic Gaussian analytic function.
This, \emph{despite} the fact that the discretization of the Stockwell transform, regular in time and in frequency, is suboptimal from an hyperbolic geometry point of view.
Furthermore, empirical estimates $\widehat{g}^{(\alpha)}$ show a good match with the theoretical pair correlation function for all four values of $\alpha$, although with a slightly larger $95\%$ credibility region for small $\alpha$.
Overall, the first and second order spatial statistics of the zeros of the discrete Stockwell Transform of white noise resemble very much the theoretical statistics of the zeros of the hyperbolic Gaussian analytic function.

\begin{remark}
Note that for all the numerical experiments presented in this paper, the sampling frequency is kept the same.
Furthermore, the proposed implementation of the discretized Analytic Stockwell Transform\footnote{\url{https://github.com/courbot/ast}} makes use of the Fast Fourier Transform algorithm and thus the frequencies on which the discrete Stockwell Transform is evaluated is \emph{fixed} as well. As a consequence, both the time and frequency bin sizes are independent of $\alpha$ which controls the size of the support of the analysis window $\varphi_{\beta}$ through $\beta = (\alpha - 1)/2$.
Consequently for very low values of $\alpha$, the time and frequency bin sizes of the discrete Analytic Stockwell Transform turn out to be inappropriate to capture correctly the time--frequency information encoded in the continuous transform. This affects the numerical retrieval of zeros and makes it less accurate which explains why the difference between the theoretical and numerical estimates of the pair correlation function increases as $\alpha$ decreases as illustrated in Figure \ref{fig:numerical_g}.
\end{remark}

These experiments show that the empirical zeros of the discrete Stockwell Transform of white noise behaves statistically very similarly compared to the zeros of the hyperbolic Gaussian analytic function. This opens the way for leveraging theses statistics to perform practical signal processing tasks: such as detection or denoising following the recent works~\cite{Flandrin2015,Abreu2018,Koliander2019,Bardenet2020,Pascal2022}.

\begin{figure}[!ht]
\centering
\subfloat[$\alpha=50$.]{\includegraphics[width=0.75\linewidth]{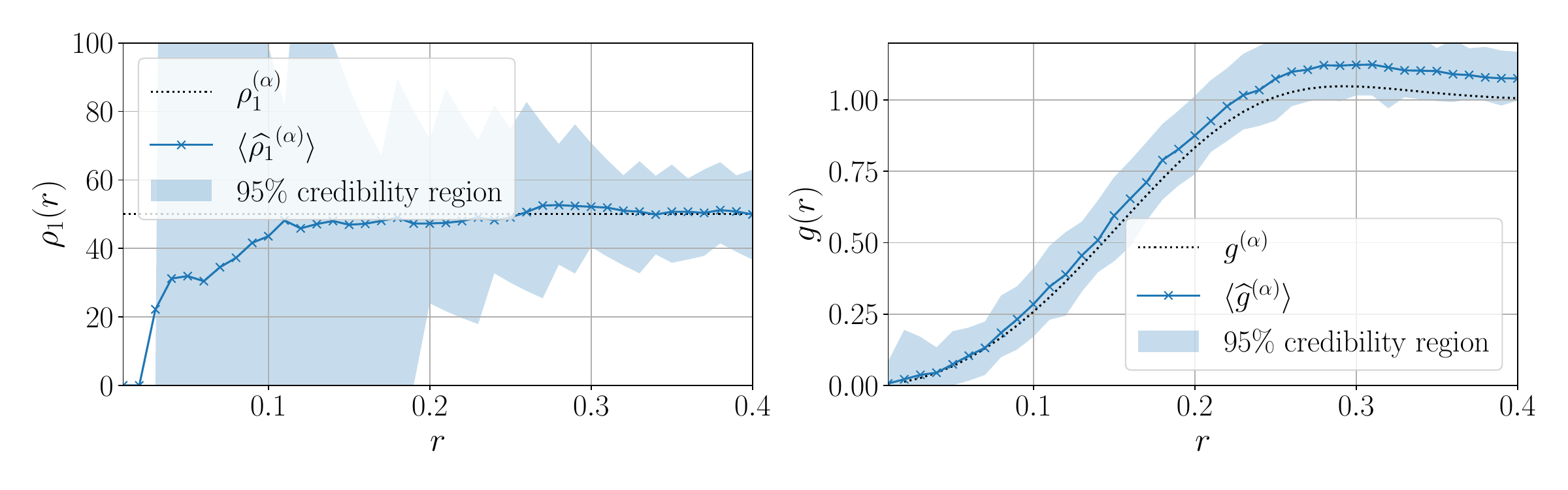}}\\
\subfloat[$\alpha=100$.]{\includegraphics[width=0.75\linewidth]{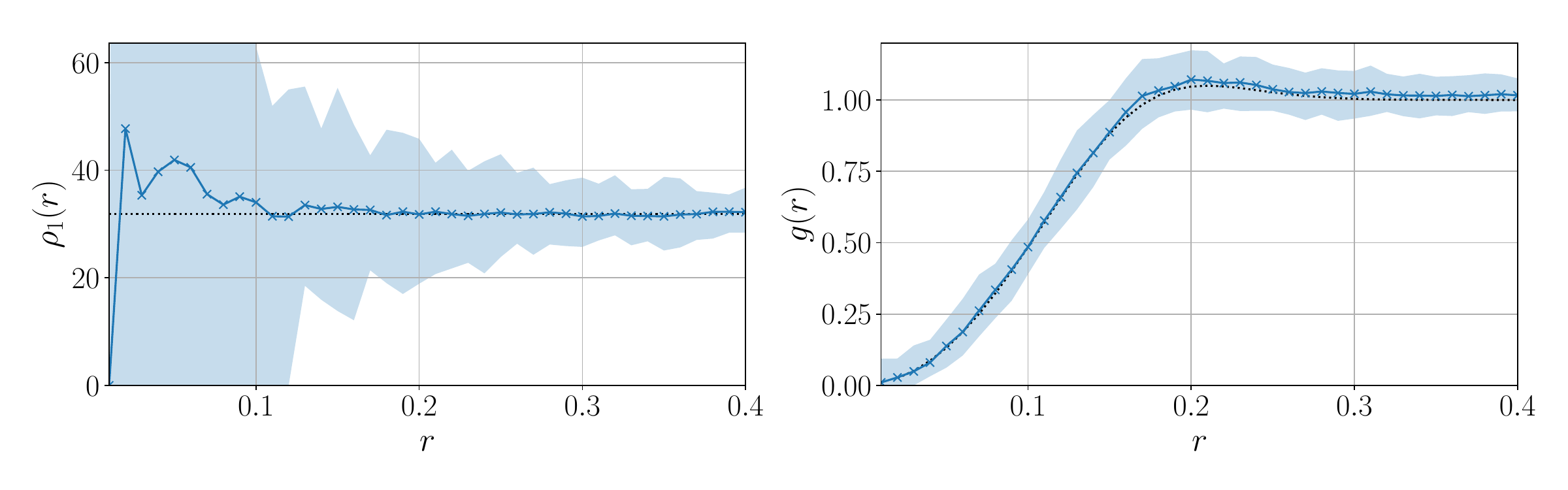}}\\
\subfloat[$\alpha=300$.]{\includegraphics[width=0.75\linewidth]{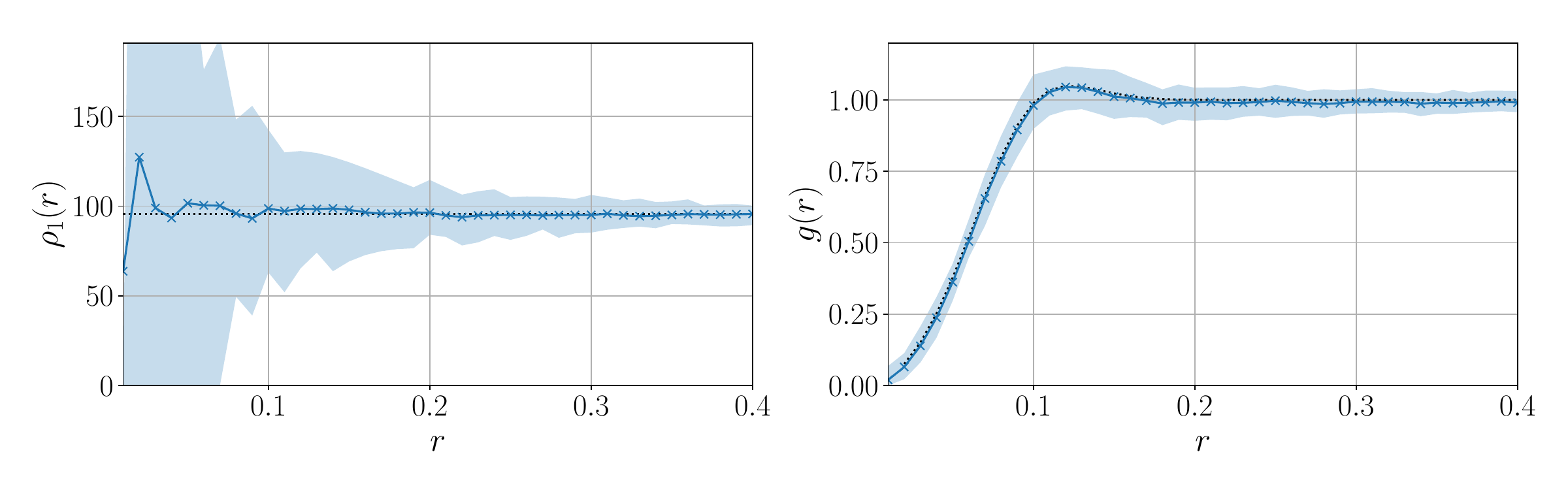}}\\
\subfloat[$\alpha=500$.]{\includegraphics[width=0.75\linewidth]{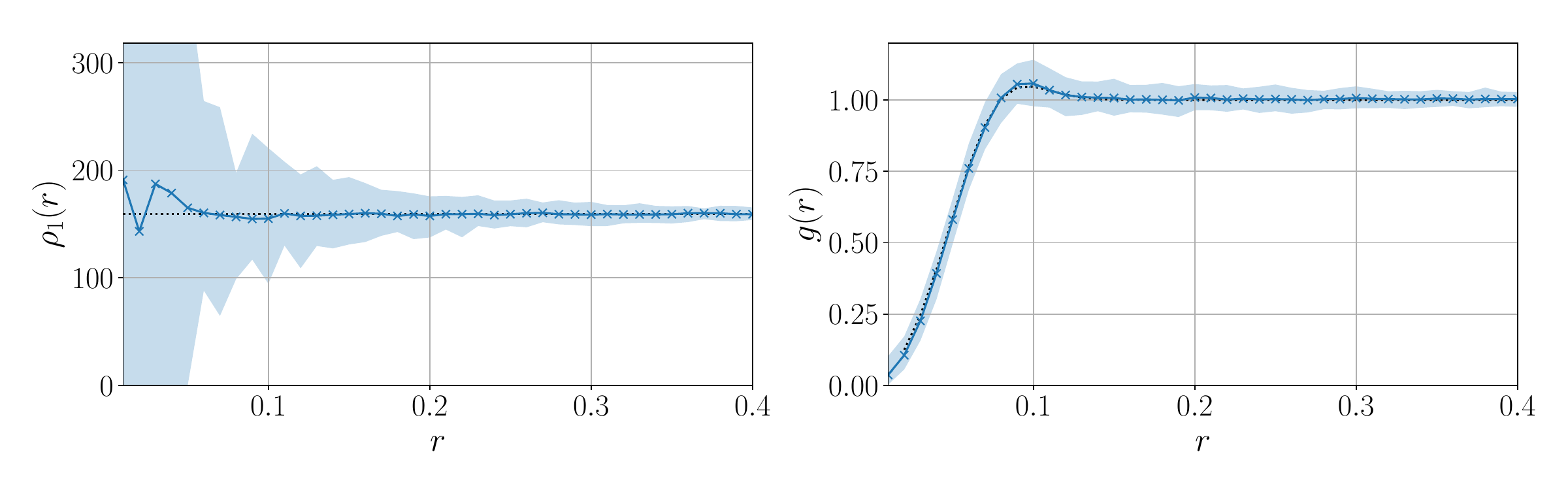}}\\
     \caption{Comparison of theoretical and empirical spatial statistics of the zeroes
     of the Analytic Stockwell Transform of parameter $\beta = (\alpha-1)/2$ of white noise for $\alpha \in \lbrace 50, 100, 300, 500\rbrace$.
     (Left) First intensity $\rho_1^{(\alpha)}$. (Right) pair correlation function $g^{(\alpha)}$.
     Theoretical spatial statistics,~\eqref{eq:rho1} and~\eqref{eq:g_th}, are represented as dashed black lines and averaged empirical statistics,~\eqref{eq:hatrho1} and~\eqref{eq:hatgalpha}, as solid blue lines.
     Ensemble averaging and quantiles estimation are performed over $R = 100$ realizations of the discrete white noise of $N= 4000$ points, corresponding to a time period of $x_{\max} - x_{\min} = 1$~s long and a sampling frequency $\nu_s = 4000$~Hz. The $x$-axis is expressed in pseudo-hyperbolic distances~\eqref{eq:dph_disk}.
     \label{fig:numerical_g}}
\end{figure}

\section{Conclusion and Perspectives}
\label{sec:conclusion}

After evidencing the existence of an Analytic Stockwell Transform, whose analysis window have been designed leveraging a modulated Cauchy wavelet, the zeros of the Analytic Stockwell Transform of white noise are characterized.
To that aim a well-suited construction of continuous white noise has been presented, enabling to establish a strong connection between the Analytic Stockwell Transform of white noise and the hyperbolic Gaussian analytic function.
This link permitted to transfer the remarkable properties of the zeros of the hyperbolic Gaussian analytic function, among which the invariance under isometries of the Poincaré disk, to the zeros of the Analytic Stockwell Transform of white noise.
Finally, intensive Monte Carlo simulations are implemented to illustrate and support this result; notably the empirical pair correlation function of the zeros of the Analytic Stockwell Transform of white noise is shown to coincide 
 accurately with the expected theoretical pair correlation function of the zero set of the hyperbolic Gaussian analytic function.
A documented Python toolbox reproducing all the experiments has been produced and made publicly available.\footnote{\url{https://github.com/courbot/ast}}

The characterization of the distribution of zeros under a white-noise hypothesis is the cornerstone of recent original zero-based detection procedures~\cite{Bardenet2020,Pascal2022,miramont2024unsupervised}.
Hence, the direct continuation of the present work would consists in leveraging the knowledge gained on the zeros Analytic Stockwell Transform to perform signal detection in high-noise contexts.
Though, application of spatial statistics in signal processing do not restrict to detection and zero-based denoising and separation strategies have shown promising results~\cite{Flandrin2015,Bardenet2020,miramont2024unsupervised}.
The nice resolution properties and great versatility of the Analytic Stockwell Transform make it a good candidate for the development of zero-based component separation algorithms. Although in this work, the {phase} information of the Stockwell Transform, which differs from the phase information contained in the Wavelet Transform, has not been exploited, it could interestingly be used to complete the information of zeros with the aim of consolidating and complementing methods for signal detection and denoising since it seems that both information are intimately related ~\cite{balazs2016pole}.
Finally, as the accuracy of zero-based procedures relies on the precise localization of zeros, novel zero detection techniques will be explored. To go beyond the Minimal Grid Neighbor algorithm, two paths will be investigated: first, an hyperbolic counterpart of the Adaptive Minimal Grid Neighbor algorithm~\cite{escudero2024efficient} will be derived, and second, the off-the-grid algorithm developed in~\cite{courbot2023sparse} will be generalized to the hyperbolic geometry framework with the overarching goal of enhancing the quality of spatial statistics estimators, and hence reaching state-of-the-art detection and reconstruction performance.

\section{Acknowledgements}

The authors thank Rémi Bardenet for fruitful discussions and references.

\clearpage
\newpage

\bibliographystyle{plain}
\bibliography{References}

\end{document}